\documentclass[12pt]{article}
\usepackage{a4wide}
\newcommand{\version}{January 19, 2013}

\usepackage{amsthm,amsfonts,amsmath, amscd}
\usepackage{amssymb,amsmath, dsfont}
\usepackage{graphicx}

\swapnumbers
\theoremstyle{plain}

\newtheorem{thm}{THEOREM}[section]
\newtheorem{lm}[thm]{LEMMA}

\theoremstyle{definition}

\theoremstyle{definition}
\newtheorem{remark}[thm]{Remark}
\newcommand{\upchi}{\raise1pt\hbox{$\chi$}}

\newcommand{\C}{{\mathord{\mathbb C}}}

\newcommand{\cH}{{\mathcal{H} }}

\newcommand{\tr}{{\rm Tr}}

\renewcommand{\|}{{\Vert}}
\numberwithin{equation}{section}
\begin{document}


\def\tr{{\rm Tr}}

\title{On an Extension Problem for  Density Matrices}
\author{\vspace{5pt} Eric A. Carlen$^1$, Joel L. Lebowitz$^{2}$  and
Elliott H. Lieb$^{3}$ \\
\vspace{5pt}\small{$1.$ Department of Mathematics, Hill Center,}\\[-6pt]
\small{Rutgers University,
110 Frelinghuysen Road
Piscataway NJ 08854-8019 USA}\\
\vspace{5pt}\small{$2.$ Departments of Mathematics and Physics,
 Hill Center,}\\[-6pt]
\small{Rutgers University,
110 Frelinghuysen Road
Piscataway NJ 08854-8019 USA}\\
\vspace{5pt}\small{$3.$ Departments of Mathematics and
Physics, Jadwin
Hall,} \\[-6pt]
\small{Princeton University, Washington Road, Princeton, NJ
  08544-0001}\\
 }
\date{\version}
\maketitle 
\footnotetext                                                                         
[1]{Work partially
supported by U.S. National Science Foundation
grant DMS 0901632  }  

\footnotetext
[2]{Work partially
supported by U.S. National Science Foundation
grant DMR  113501  and AFOSR FA9550-10-1-0131}

\footnotetext
[3]{Work partially
supported by U.S. National Science Foundation
grant PHY 0965859  and a grant from the Simons Foundation (\# 230207 to
Elliott Lieb).  \\
\copyright\, 2012 by the authors. This paper may be
reproduced, in its
entirety, for non-commercial purposes.}

\begin{abstract}
We investigate the problem of the existence of a density matrix
$\rho_{123} $ 
 on  a Hilbert space $\cH_{1}\otimes\cH_2\otimes
\cH_3$ with given partial traces 
$\rho_{12} =\tr_{3} \, \rho_{123}$ and $\rho_{23}=\tr_{1}\, \rho_{123}$.
While we do not solve this problem completely
we offer partial results in the form of some necessary and some sufficient
conditions on $\rho_{12}$ and   $\rho_{23}$. The quantum case differs
markedly from the classical (commutative) case, where the obvious necessary
compatibility condition suffices, namely, $\tr_1\, \rho_{12} = \tr_3
\rho_{23}$.

\end{abstract}

\medskip
\leftline{\footnotesize{\qquad Mathematics subject
classification numbers: 81V99, 82B10, 94A17}}
\leftline{\footnotesize{\qquad Key Words: Density matrix, Entropy, Partial
trace }}

\section{Introduction} \label{intro}

The problem considered here is closely related to the {\em quantum marginal problem}, on which there is an extensive literature. 
However, since the problem we consider involve {\em overlapping marginals}, results in the literature shed little light on it.
We therefore introduce the problem in terms that we find natural, and postpone  the discussion of the relation between our results and
results on the quantum marginal problem until later in the introduction. 

Let $\cH_1, \, \cH_2, \, \cH_3 \, $ be three finite dimensional Hilbert
spaces. Let $\rho_{12} $ be a density matrix on $\cH_{12} = \cH_1\otimes
\cH_2$ and, similarly, let $\rho_{23} $ be a density matrix on $\cH_{23}
= \cH_1\otimes
\cH_2$.  The question we ask, and which we can only partially resolve, is:

{\em Assuming that there is consistency of the partial traces, namely
$\tr_1\, \rho_{12} = \rho_2 = \tr_3\, \rho_{23}$, what are necessary and
sufficient conditions for the existence of a density matrix $\rho_{123}$
on $\cH_{123} = \cH_1\otimes \cH_2\otimes \cH_3$ such that 
$\tr_1\,\rho_{123} = \rho_{23} $ and $\tr_3 \, \rho_{123} = \rho_{12} $?}

\medskip
This is an obvious question to ask in several contexts, e.g., \cite{kls}.
We note the fact that in classical statistical mechanics the
answer is that an extension always exists, and there is a simple formula
for it:  Given probability densities\footnote[1]{If the sample spaces over which $x$, $y$ and $z$  range are finite sets, these densities
may be identified with diagonal density matrices in an natural way, embedding the
discrete classical probability space problem into the quantum mechanical problem} $\rho_{12}(x,y)$ and 
$\rho_{23}(y,z)$ with
$$\int \rho_{12}(x,y){\rm d}\mu_1(x) = \int \rho_{23}(y,z){\rm d}\mu_3(z) =:
\rho_2(y)\ ,$$ we may define
\begin{equation}\label{constr}
 \rho_{123}(x,y,z) = \frac{\rho_{12}(x,y)\rho_{23}(y,z)}{\rho_2(y)}\ .
\end{equation}

This extension is not  unique, in general, but among all extensions, this
one has the maximum entropy, namely $S_{12} + S_{23} - S_2$. The maximality
is a consequence of the the Strong Subadditivity of Entropy (SSA) which says
that
for any extension
$$S_{12} + S_{23} \geq S_{123}+  S_2\ .$$
The classical SSA inequality is relatively easy to prove (as opposed to its quantum version); see \cite{l}.
The principle behind the construction is {\em conditioning}: Note that   for random variables $Y$ and $Z$
with joint density  $\rho_{23}(y,z)$, ${\rho_{23}(Y,z)}/{\rho_2(Y)}$ is the conditional density for $Z$ given $Y$.

This construction of extensions by conditioning can be generalized to arbitrarily many factors. 
Given $\rho_{j,j+1}$, $j=1,\dots,N$, such that 
$$\int \rho_{j-1,j}(x_{j-1},x_j){\rm d}\mu_{j-1}(x_{j-1}) = \int
\rho_{j,j+1}(x_j,x_{j+1}){\rm d}\mu_{j+1}(x_{j+1}) =: \rho_{j}(x_j)\ ,$$
define
$$\rho_{1,\dots,N}(x_1,\dots,x_N) =
\rho_{1,2}(x_1,x_2)\prod_{j=2}^{N-1}\frac{\rho_{j,j+1}(x_j,x_{j+1})}{
\rho_j(x_j) }\ .$$
Again this extension is the maximum entropy extension by an interated
application of SSA.

However, the construction (\ref{constr}), being based on conditional probabilities,  does not generalize to the quantum
case, where the notion of ``conditioning'' has no obvious meaningful analog. 
In general, even for consistent density matrices,
$$\rho_{12}\rho_{2}^{-1}\rho_{23}$$
need not be Hermitian, much less positive definite. While
$$
R=\exp[\log\rho_{12} + \log\rho_{23} - \log\rho_2]
$$
is Hermitian and positive definite its partial traces will not equal the
desired reduced density matrices, even after renormalization, which would
be required since the trace, $\tr_{123} R$, is never
greater than 1, and, generally, is less than 1. This fact follows from
the triple Golden-Thompson inequality \cite{wy}: 
$$
\tr_{123}\,  R \leq \tr_2\tr_{13}\int_0^\infty \rho_{12}\, \frac{1}{t +
\rho_2} \, \rho_{23} \, \frac{1}{t +
\rho_2}\, {\rm d}t = \tr_2 \int_0^\infty \rho_{2} \, \frac{1}{t +
\rho_2} \, \rho_{2}\,   \frac{1}{ t +
\rho_2} \, {\rm d}t = \tr_2 \rho_2 = 1 \ .
$$

There are, in fact, consistent pairs of density matrices $\rho_{12}$ and
$\rho_{23}$ that have {\em no} extension:  Suppose that $\rho_{12}$ is pure.
If  $\rho_{123}$ is such that $\tr_3 \rho_{123} = \rho_{12}$, the
purity of $\rho_{12}$ forces $\rho_{123}$ to have the form $\rho_{12}\otimes\rho_3$
and hence $\rho_{23} = \rho_{2}\otimes \rho_{3}$.

Thus, if $\rho_{12}$ is pure, the only compatible density matrices
$\rho_{23}$ with which
it has a common extension are the products  $\rho_{2}\otimes \rho_{3}$,
in which case the unique common extension is $\rho_{12}\otimes\rho_{3}$. 
Theorem~\ref{prod only} in the next section generalizes this result by identifying
a class of non-pure states $\rho_{12}$,
which can only be extended by product states $\rho_{12}\otimes
\rho_3$.

Motivated by these examples, and the obvious difference between the
classical and quantum cases, we believe that a proper understanding of this
problem will lead to a clearer understanding of entanglement and quantum
information theory. 

\subsection{The quantum marginal problem}

The quantum marginal problem, in its simplest form,  is the following: Given density matrices $\rho_1$ and $\rho_2$, on $\cH_1$ and $\cH_2$ respectively, 
let $\mathcal{C}(\rho_1,\rho_2)$ denote the set of all density matrices $\rho_{12}$ on $\cH_1\otimes \cH_2$
such that 
\begin{equation}\label{two}
\tr_1 \rho_{12} =\rho_2\qquad{\rm  and}\qquad\tr_2 \rho_{12} =\rho_1\ .
\end{equation}
$\mathcal{C}(\rho_1,\rho_2)$ is the set of {\em quantum couplings} of $\rho_1$ and $\rho_2$. 

Note that $\mathcal{C}(\rho_1,\rho_2)$ is never empty; it always contains $\rho_1\otimes \rho_2$. It is also 
evidently convex and compact, and hence it is the convex hull of its extreme points. 
Characterizations of the extreme points have been given by Parathasarathy \cite{P} and Rudolph \cite{R}.

Other research has focused on the relations between the spectra of $\rho_1$, $\rho_2$ and $\rho_{12}$ that
characterize the set of triples of density matrices satisfying (\ref{two}). For results in this direction see the recent papers of Klyachko \cite{K},
and Christandl et. al. \cite{CDKW}, and references therein.

The general quantum marginal problems concerns states on the product of arbitrarily many Hilbert spaces, and the marginals obtained by
taking any combination of partial traces -- for instance,  the partial traces $\rho_{12}$ and $\rho_{23}$ of $\rho_{123}$ as in our problem.
Most results pertain to the case of non-overlapping marginals, in contrast to the problem considered here. An exception is
the paper by Osborne \cite{O}.  He uses the SSA of the von Neumann entropy \cite{LR} to prove an upper bound the number of {\em orthogonal
pure states} $\rho_{123}$ such that $\tr_1\,\rho_{123} = \rho_{23} $ and $\tr_3 \, \rho_{123} = \rho_{12} $. We shall also employ entropy bounds, but are 
mainly concerned with the existence, or not, of {\em mixed-state} extensions $\rho_{123}$ of compatible pairs $\rho_{12}$ and $\rho_{23}$.

It is natural to use entropy in this investigation, and some early results  on the two-space problem are given in terms of entropy. 
For example, if $\rho_{12}$ is a pure state, then its partial traces $\rho_1$ and $\rho_2$ have the same non-zero spectrum \cite{al}.
Conversely, if  $\rho_1$ and $\rho_2$ are two density matrices having the same non-zero spectrum, then there is a pure state
$\rho_{12}$ satisfying (\ref{two}). To construct it, assume without loss of generality that $\cH_1$ and $\cH_2$ have the same dimension. 
There is a unitary $U$ such that $U\rho_2U^* = \rho_1$. Define $\Psi :=
U\sqrt{\rho_2}$, where the matrix on the right is
regarded as a vector in $\cH_1\otimes \cH_2$. 
Then $\rho_{12} = |\Psi\rangle\langle \Psi|$ is such a pure state. 

If $\rho_1$ and $\rho_2$ do not have the same spectrum, then the set $\mathcal{C}(\rho_1,\rho_2)$ cannot contain any pure state.
It is natural then to ask for the least entropy element of  $\mathcal{C}(\rho_1,\rho_2)$.  Since the entropy is concave, the minimum will
be attained at an extreme point, and the results of \cite{P,R} are relevant,
and lead easily to the minimizer in specific cases. However, a
well-known inequality already provides 
a sharp {\em a-priori} lower bound for this minimum entropy coupling:

The {\em Araki-Lieb Triangle inequality}  \cite{al} says that when
(\ref{two}) is satisfied, 
\begin{equation}\label{alt}
S_{12} \geq |S_1 - S_2|\ .
\end{equation}
For an interesting discussion of entropy inequalities for a general {\em classical} marginal problem, see \cite{FC}.

\section{Necessary Conditions for the Existence of an Extension}

Strong Subadditivity of Entropy (SSA) \cite[see also \cite{l}]{LR} provides
us with two necessary conditions for an extension to exist:
Let $\rho_{123}$ be an extension of $\rho_{12}$ and $\rho_{23}$. Then,
discarding the positive term $S_{123}$ in SSA, 
\begin{equation}\label{cheap}
 S_{12}+ S_{23} \geq S_2\ .
\end{equation}
Classically, one has monotoniciy of the entropy, meaning $S_{12} \geq S_2$
and
$S_{23} \geq S_2$ so that (\ref{cheap}) hold classically with $2S_2$ on the
right side. However, quantum mechanically, equality can hold in
(\ref{cheap}): Let $\rho_{12}$ be a purification
 of $\rho_2$, i.e., a pure state $|\Psi\rangle\langle\Psi|$ on $\cH_1\otimes \cH_2$ 
such that $\tr_2(|\Psi\rangle\langle\Psi|) = \rho_1$. Such purifications always exist\footnote[2]{It is well known, and easy to see from the definition, that the set of all
possible purifications $\Psi$ of $\rho_1$ is  the set of all operators $\sqrt{\rho_1}U$, 
regarded as vectors in $\cH_1\otimes \cH_2$, where $U$ is a partial isometry from a subspace of 
$\cH_2$ onto the range of $\rho_1$.}.  Let
$\rho_{23}$ be the tensor product of $\rho_2$ and a pure state. Then
$S_{12} =0$ and $S_{23} = S_2$. 

A second form of SSA \cite{LR} leads to a sharper necessary condition: 
Again
assume that $\rho_{123}$ is an extension of $\rho_{12}$ and $\rho_{23}$.
Then 
\begin{equation}\label{pol}
S_{12} + S_{23} \geq S_1 +S_3\ . 
\end{equation}
As we now explain, whenever (\ref{pol}) is satisfied by any consistent pair
of density matrices $\rho_{12}$ and $\rho_{23}$ (not necessarily
possessing a common extension), then (\ref{cheap}) is automatically
satisfied. 
To see this, note that by (\ref{alt}), $S_2 - S_1 \leq S_{12}$ and $S_2 - S_3 \leq S_{23}$. Adding these
inequalities, we obtain
$$2S_2 \leq S_{12}+S_{23} + S_1+ S_3\ .$$
By (\ref{pol}), the right side is no greater than $2(S_{12}+S_{23})$, which
yields (\ref{cheap}).

The density matrices for which there is equality in the triangle inequality
have a particular structure:
Let $m$ and $n$ be positive integers, and let $\{\lambda_1,\dots,\lambda_m\}$ and $\{\mu_1,\dots,\mu_n\}$
be sets of positive numbers with $\sum_{j=1}^m\lambda_j = \sum_{k=1}^n\mu_k =1$. 
Then, as shown in \cite{cl}, there exists a density matrix $\rho_{12}$ such that the non-zero eigenvalues of $\rho_{12}$ are 
$\{\lambda_1,\dots,\lambda_m\}$, the non-zero eigenvalues of $\rho_{2}$ are $\{\mu_1,\dots,\mu_n\}$
and the non-zero eigenvalues of $\rho_1$ are the numbers $\{\lambda_j\mu_k\ :\ 1\leq j \leq m\ ,\ 1 \leq k \leq n\}$. 
For any such $\rho_{12}$, it is evident that $S_{12}= S_1 - S_2$. Moreover, as shown in \cite{cl}, whenever 
$S_{12}= S_1 - S_2$, the spectra of $\rho_{12}$, $\rho_{1}$ and $\rho_2$ are related in this way. 

The following theorem generalizes the observation that pure states
$\rho_{12}$
may only be extended by product states $\rho_{12}\otimes \rho_3$.  When
$\rho_{12}$
is pure, $0 = S_{12} = S_2-S_1$.

\begin{thm}\label{prod only} Let $\rho_{12}$ be a density matrix such that
\begin{equation}\label{tri1}
S_{12} = S_1- S_2\ .
\end{equation}
Then $\rho_{12}$ and $\rho_{23}$ have a common extension if and only if 
$\rho_{23}= \rho_{2}\otimes \rho_3$.
\end{thm}

\begin{proof} Using (\ref{tri1}) in (\ref{pol}) we obtain
\begin{equation*}
S_{1} - S_{2} + S_{23} \geq S_{3} + S_1\ .
\end{equation*}
which reduces to $S_{23} \geq S_2+S_3$. By the subadditivity of the entropy, this means that
$S_{23} = S_2+S_3$, and so $\rho_{23} = \rho_2\otimes \rho_3$. 
\end{proof}

In Section 4 we give an example in which (\ref{pol}) is satisfied, but
there is no common extension.

\section{Sufficient Conditions for the Existence of an Extension}

We do not have any very general sufficient conditions for compatible pairs to possess a common extension. One general positive statement that can be made is
the following:

\begin{thm} Let $\rho_{12}$ and $\rho_{23}$ be a compatible pair of density matrices  that  posses a common extension $\rho_{123}$
that is positive definite.  Let $\|\cdot\|$ denote the trace norm. Then there is an $\epsilon>0$, depending on the dimensions and the smallest eigenvalue of $\rho_{123}$, 
such  that if $\widetilde \rho_{12}$ and  $\widetilde \rho_{23}$ is another compatible pair
on the same spaces and 
\begin{equation}\label{open}
\|\rho_{12} - \widetilde\rho_{12}\| +  \|\rho_{23} - \widetilde\rho_{23}\| < \epsilon\ ,
\end{equation}
then  $\widetilde \rho_{12}$ and  $\widetilde \rho_{23}$ possess common extension.
\end{thm}

\begin{proof}

Define  
$$\widetilde \rho_{123} = \rho_{123} + [\widetilde \rho_{12} - \rho_{12}]\otimes \widetilde \rho_3 + \widetilde \rho_1\otimes [\widetilde \rho_{23} - \rho_{23}]
+ \widetilde \rho_1\otimes[ \rho_2 - \widetilde \rho_2]\otimes \widetilde \rho_3\ .$$
It is easily checked that $\tr_1\widetilde \rho_{123} = \widetilde \rho_{23}$ and   $\tr_3\widetilde \rho_{123} = \widetilde \rho_{12}$.  Furthermore, 
$\widetilde \rho_{123}$ is self-adjoint, and under the condition (\ref{open}), is positive when $\epsilon$ is chosen sufficiently small.  
\end{proof}

\begin{remark} Of course, in this finite dimensional setting, the trace norm could be replaced by any other norm.  Also,
Suppose that $\rho_{12}$ is positive definite, and $\rho_{23} = \rho_2\otimes \rho_3$ with $\rho_3$ positive definite. In this case
$\rho_{12}$ and $\rho_{23}$ are compatible and have the positive-definite  common extension $\rho_{12}\otimes \rho_3$. 
The theorem says that any compatible pair that is a small perturbation of  $\rho_{12}$ and $\rho_{23}$ has a common extension. 
However, as Theorem~\ref{prod only} shows, the requirement of positive definiteness cannot be dropped. 
\end{remark}

Let $\rho_{12}$ and $\rho_{23}$ be a consistent pair of density matrices,
both of which are separable.  Recall that a bipartite density
matrix $\rho_{12}$ is {\em finitely separable} if is is a convex combination of product states; i.e.,
if
$$\rho_{12} = \sum_{j=1}^N \lambda_j \sigma_1^{(j)}\otimes  \tau_2^{(j)}$$
where each $\sigma_1^{(j)}$ is a density matrix on $\cH_1$, each 
$\tau_2^{(j)}$ is a density matrix on $\cH_2$, and, crucially, each $\lambda_j > 0$. 
The set of separable density matrices is the closure of the set of finitely separable density matrices. 

Separable density matrices are often viewed as being 
a classical ensemble of product states. As such, separable states are free
of {\em quantum correlations} among observables on $\cH_1$ and $\cH_2$.
Separable bipartite density matrices
do, indeed,  behave much more like classical joint probability
distributions. For example, if $\rho_{12}$ is separable,
then (see e.g. \cite{cl})
$$S_{12} \ge \max\{S_1\, \ S_2\}\ ,$$
and thus when both $\rho_{12}$ and $\rho_{23}$ are separable and compatible, $S_{12}+S_{23} \geq S_1+S_3$,
and thus the condition (\ref{pol}), which is necessary for a common  extension to exists, is always satisfied. 

One might therefore, hope that a common extension
$\rho_{123}$ would exist whenever $\rho_{12}$ and $\rho_{23}$ are compatible
and both are separable. 
In the next  section, we show that this is not the case: Further conditions must be imposed to ensure the existence of an extension. 
Here is one case in which a common extension does exist:

Suppose  that
\begin{equation}\label{sep1}
\rho_{12} = \sum_{j=1}^n \lambda_j \, \rho^{(j)}\otimes \sigma^{(j)}
\qquad{\rm and}\qquad
\rho_{23} = \sum_{j=1}^n \lambda_j \, \sigma^{(j)}\otimes \tau^{(j)}\,
\end{equation}
where the $\rho^{(j)}$, $\sigma^{(j)}$ and $\tau^{(j)}$ are density matrices
and the $\lambda_j$ are positive numbers with $\sum_{j=1}^n\lambda_j = 1$. 
Then
\begin{equation}\label{sep2}
 \rho_{123} :=  \sum_{j=1}^n \lambda_j \, \rho^{(j)}\otimes
\sigma^{(j)}\otimes \tau^{(j)}
\end{equation}
is a common extension. 

While this example is based on a strong assumption, namely
that the weights and the factors on $\cH_2$ coincide, we shall show in Section 4
that consistency and separability is not enough. 

However, as we now explain, there is a more general version of this
construction using coherent states. Good references on coherent states are
\cite{KS,AP}.

Let $\cH$ be any $d$ dimensional Hilbert space. Define $J = (d-1)/2$. 
Then there is an irreducible representation of $SU(2)$ on $\cH$.  Associated to this representation is a family
of pure states $|\Omega\rangle\langle \Omega|$ on $\cH$, parameterized by 
a point $\Omega$ in the unit sphere, $S^2$. 
Given any self-adjoint operator $A$ on $\cH$, there is a function
$\widehat a(\Omega)$ on $S^2$ such that
\begin{equation}\label{upper}
 A = \int_{S^2} {\rm d}\Omega\, \, \widehat
a(\Omega)\, |\Omega\rangle\langle \Omega|\
\end{equation}
where the integration is with respect to the uniform probability measure on
$S^2$. 
The function $\widehat a(\Omega)$ is {\em unique} if we further require that
$\widehat a(\Omega)$ be a spherical harmonic of degree no higher than
$2J$. (See \cite[Page 33 and Equation (4.10)]{KS}.)   
This function $\widehat a(\Omega)$ is called the {\em upper symbol} of the
operator $A$. 
The upper symbol of $A$ need not be a non-negative function on $S^2$ even if
$A$ is positive definite. However, if $\widehat a$ is non-negative, then
the operator $A$ defined by (\ref{upper}) is positive semidefinite.

Now consider a bipartite density matrix $\rho_{12}$ on $\cH_1\otimes \cH_2$ 
where the dimension of $\cH_j = d_j$ for $j=1,2$. Suppose that has a
representation
$$
\rho_{12} = \int_{S^2}{\rm d}\Omega_2\int_{S^2}{\rm d}\Omega_1 \,
\widetilde \rho_{12}(\Omega_1,\Omega_2) \,  
\left(|\Omega_1\rangle\langle \Omega_1| \otimes 
|\Omega_2\rangle\langle \Omega_2|\right)
$$
with
$\widetilde \rho_{12}(\Omega_1,\Omega_2)$ a nonnegative function on
$S^2\times S^2$ that is a spherical harmonic of  degree $d_j-1$ in $\Omega_j$, $j=1,2$.

Then taking the trace over $\cH_1$, we find
$$\rho_2 := \tr_1\rho_{12} =   \int_{S^2}{\rm d}\Omega_2 \left[
\int_{S^2}{\rm d}\Omega_1\, \,  \widetilde\rho_{12}(\Omega_1,\Omega_2)\,
\right] |\Omega_2\rangle\langle \Omega_2| \ ,$$
and evidently 
${\displaystyle  \int_{S^2}  {\rm
d}\Omega_1}\, \widetilde \rho_{12}(\Omega_1,\Omega_2)\,  
$
has the degree $d_2-1$ in $\Omega_2$,  which ensures that it is the unique upper symbol of
$\rho_{2}$
of minimal degree. 

Likewise, if $\rho_{23}$ is a density matrix on $\cH_2\otimes \cH_3$ with the dimension of $\cH_3 = d_j$, 
and if $\rho_{23}$ is compatible with $\rho_{12}$ and moreover
$$\rho_{23} = \int_{S^2}{\rm d}\Omega_2\int_{S^2}{\rm d}\Omega_3\, \, 
\widetilde\rho_{23}(\Omega_2,\Omega_3) \, \left(|\Omega_2\rangle\langle
\Omega_2|  \otimes|\Omega_3\rangle\langle \Omega_3| \right)
$$
with
$\widetilde\rho_{23}(\Omega_2,\Omega_3)$ being  a nonnegative function on
$S^2\times
S^2$ that is a spherical harmonic of  degree $d_j-1$ in $\Omega_j$, $j=2,3$, 
$\int_{S^2}{\rm
d}\Omega_3\,\widetilde\rho_{23}(\Omega_2,\Omega_3)$ is the upper
symbol of $\rho_2$, and
thus quantum compatibility of $\rho_{12}$ and $\rho_{23}$ implies the
classical compatibility of the probability densities
$\widetilde\rho_{12}(\Omega_1,\Omega_2)$ and $\widetilde\rho_{23}(\Omega_2,\Omega_3)$.
Hence the classical prescription may be used to define
$$\widetilde\rho_{123}(\Omega_1,\Omega_2,\Omega_3) :=
\frac{\rho_{12}(\Omega_1,\Omega_2)\, \rho_{23}(\Omega_2,\Omega_3)}{
\rho_2(\Omega_2) } \ ,$$
and it is then the case that
$$\rho_{123} := \int_{S^2}{\rm d}\Omega_1 \int_{S^2}{\rm d}\Omega_2
\int_{S^2}{\rm d}\Omega_3\, \, \widetilde
\rho_{123}(\Omega_1,\Omega_2,\Omega_3)\,
\left(|\Omega_1\rangle\langle \Omega_1| \otimes 
|\Omega_2\rangle\langle \Omega_2| \otimes 
|\Omega_3\rangle\langle \Omega_3|\right)
$$
is positive semidefinite, and 
is a common extension of $\rho_{12}$ and $\rho_{23}$.

\section{Separable Compatible Pairs with no Extension}

In this section we give examples of compatible pairs $\rho_{12}$ and
$\rho_{23}$ that do not have an extension, but nevertheless satisfy the 
necessary condition \eqref{pol}. Moreover, in these examples both
$\rho_{12}$ and $\rho_{23}$ will be  separable.
As we have remarked above,  one might expect the extension
problem to simplify in the presence of separability. This is not the case, as the following examples show. 
For simplicity, our
examples will be on $\C^2\otimes\C^2\otimes\C^2$, but of course may be embedded into higher dimensional spaces. 

\begin{lm}
 Let $\rho_{2}$ be a non-pure density matrix on $\C^2$; i.e., the rank
of $\rho_2$ is $2$. Let $\{\psi_1,\psi_2\}$
 be an orthonormal basis of $\C^2$ consisting of eigenvectors of $\rho_2$. Let 
 $\rho_2\psi_j = \mu_j\psi_j$, $j=1,2$. 
  Then there
exist two  unit vectors  $\phi_1$ and $\phi_2$ in $\C^2$
and positive numbers  $\nu_1$ and $\nu_2$ such that
\begin{equation}\label{twoB}
 \sum_{j=1}^2 \mu_j |\psi_j\rangle\langle \psi_j| = \rho_2 =
\sum_{j=1}^2
\nu_j |\phi_j\rangle\langle \phi_j|\ ,
\end{equation}
and such that the four vectors $\psi_1$, $\psi_2$, $\phi_1$ and $\phi_2$ are pairwise linearly independent. 
\end{lm}

\begin{proof}  If $\rho_2$ is $\tfrac12{\mathds 1}$, we may take
$\{\psi_1,\psi_2\}$
to be any orthonormal basis of $\C^2$, and then choose $\{\phi_1,\phi_2\}$
to be any other orthonormal basis such that $\phi_1$ is not proportional to 
either $\psi_1$ or $\psi_2$. 
Then 
the four vectors
$\phi_1,\phi_2,\psi_1,\psi_2$  are {\em pairwise} linearly independent.
In this case, we take $\mu_j = \nu_j =\tfrac12$ for $j=1,2$. 

Next, assume $\rho_2$ has distinct eigenvalues $\mu_1>\mu_2$.
Let $\{\psi_1,\psi_2\}$ be an orthonormal basis consisting
of eigenvectors of $\rho_2$. Let $\phi_1$ be any unit vector that is
not proportional to either $\psi_1$ or $\psi_2$. Then there is a unique
largest number $\nu_1>0$ so that $\rho_2 - \nu_1|\phi_1\rangle\langle
\phi_1|$ is positive semidefinite. Thus,
$\rho_2 - \nu_1|\phi_1\rangle\langle
\phi_1|$ is rank one, and hence
$$\rho_2 - \nu_1|\phi_1\rangle\langle 
\phi_1|= \nu_2 |\phi_2\rangle\langle
\phi_2|$$
for some unit vector $\phi_2$ that is not proportional to $\phi_1$ since
$\rho_2$ has rank $2$. Likewise, $\phi_2$ is not an eigenvector of
$\rho_{2}$, since by the non-degeneracy, this would force $\phi_1$ to be an 
eigenvector, which it is not. 
\end{proof}

\begin{thm}\label{noext}
 Let $\rho_{2}$ be a non-pure density matrix on $\C^2$. Then there exist
separable density matrices $\rho_{12}$ and $\rho_{23}$,  extending $\rho_2$,
 and such  that
\eqref{pol} is satisfied (with strict inequality), but such that
$\rho_{12}$ and $\rho_{23}$ have no common extension $\rho_{123}$.
\end{thm}

\begin{proof}
 Let $\{\psi_1,\psi_2\}$ and $\{\phi_1,\phi_2\}$ be sets of
unit vectors such that \eqref{two} is satisfied, and where $\{\psi_1,\psi_2\}$ is an orthonormal basis of $\C^2$ consisting of eigenvectors
of $\rho$.

Let $\{\chi_1,\chi_2\}$ be any  orthonormal
basis of $\C^2$, and let $\{\eta_1,\eta_2\}$ be any linearly independent set of unit vectors in $\C^2$.  
Define
\begin{equation}\label{def}
\rho_{12} = \sum_{j=1}^2 \mu_j | \eta_j\otimes \psi_j\rangle \langle \eta_j\otimes \psi_j|  \qquad{\rm and}\qquad
\rho_{23} = \sum_{j=1}^2 \nu_j | \phi_j\otimes \chi_j\rangle \langle \phi_j\otimes \chi_j|\ .
\end{equation}
Then evidently $\tr_1\, \rho_{12} = \tr_{3}\, \rho_{23} = \rho_2$. 
Next, for any vector $(z,w)\in \C^2$, let $(z,w)^\perp = (-\overline{w},\overline{z})$ so that $(z,w)^\perp$
is orthogonal to $(w,z)$.

Suppose a common extension $\rho_{123}$ does exist. 
Then $\eta_1^\perp \otimes \psi_1$ is in the nullspace of $\rho_{12}$, and hence
$$\sum_{j=1}^2 \langle \eta_1^\perp \otimes \psi_1\otimes \chi_j, \rho_{123}\,
 \eta_1^\perp \otimes \psi_1\otimes \chi_j\rangle = 
 \langle \eta_1^\perp \otimes \psi_1, \rho_{12}\, \eta_1^\perp \otimes \psi_1\rangle = 0\
.$$
Therefore, since $\rho_{123}$ is positive, both
$\eta_1^\perp \otimes \psi_1\otimes \chi_1$ and
$\eta_1^\perp \otimes \psi_1\otimes \chi_2$ are in the nullspace of $\rho_{123}$.

Since $\eta_2^\perp \otimes \psi_2$ is in the null space of $\rho_{12}$, the
same  argument shows that 
both
$\eta_2^\perp \otimes \psi_2\otimes \chi_1$ and
$\eta_2^\perp \otimes \psi_2\otimes \chi_2$ are in the nullspace of $\rho_{123}$. 
Thus, the four vectors
$$\eta_1^\perp \otimes \psi_1\otimes \chi_j \qquad{\rm and}\qquad \eta_2^\perp \otimes \psi_2\otimes \chi_j\ ,  \qquad j = 1,2$$
belong to the nullspace of $\rho_{123}$.

Likewise,  both $\phi_!^\perp\otimes \chi_1$ and $\phi_2^\perp \otimes \chi_2$
belong to the nullspace of $\rho_{23}$. 
Arguing as above, we see that  the four vectors
$$\eta_j^\perp \otimes \phi_!^\perp\otimes \chi_1 \qquad{\rm and}\qquad \eta_j^\perp \otimes \phi_2^\perp \otimes \chi_2\ ,  \qquad j = 1,2$$
belong to the nullspace of $\rho_{123}$. 

Define the unit vectors
\begin{equation}\label{psidef}
\Psi_1 = \eta_1^\perp \otimes \psi_1\otimes \chi_1 \ , \quad  \Psi_2 = \eta_1^\perp \otimes \psi_1\otimes \chi_2 \ , \quad  
\Psi_3 = \eta_2^\perp \otimes \psi_2\otimes \chi_1 \ , \quad {\rm and}\qquad \Psi_4 = \eta_2^\perp \otimes \psi_2\otimes \chi_2 \ , \quad 
\end{equation}
and the vectors 
\begin{equation}\label{phidef}
\Phi_1 = \eta_1^\perp \otimes \phi_1^\perp\otimes \chi_1 \ , \quad  \Phi_2 = \eta_1^\perp \otimes \phi_2^\perp\otimes \chi_2 \ , \quad  
\Phi_3 = \eta_2^\perp \otimes \phi_1^\perp\otimes \chi_1 \ , \quad {\rm and}\qquad \Phi_4 = \eta_2^\perp \otimes \phi_2^\perp\otimes \chi_2 \ , \quad 
\end{equation}
These $8$ vectors are in the nullspace of $\rho_{123}$under the hypotheses imposed so far. 

Now let us temporarily impose the additional hypothesis that $\{\eta_1,\eta_2\}$ is orthonormal. Consequently, 
$\{\eta_1^\perp,\eta_2^\perp\}$ is orthonormal. 

We claim that the $8$ unit vectors $\Psi_j,\Phi_j$, $j=1,2,3,4$,  are linearly independent. To see this, note that under our additional hypothesis,
 $\{\Psi_1,\Psi_2,\Psi_3,\Psi_4\}$ and
$\{\Phi_1,\Phi_2,\Phi_3,\Phi_4\}$ are orthonormal, and $\langle \Psi_i,
\Phi_j\rangle = \delta_{i,j}$. 
Therefore, if 
$$\sum_{j=1}^4 a_j\Psi_j +  \sum_{j=1}^4 b_j\Phi_j = 0\ ,$$
$$(a_j\Psi_j + b_j\Phi_j) = 0\qquad{\rm for\ each}\qquad j=1,2,3,4. $$
However, since the vectors $\psi_1$ and $\phi_1$ are linearly independent,
so are the vectors  $\psi_1^\perp$ and $\phi_2= \phi_1^\perp$. 
Thus $\Psi_1$ and $\Phi_1$ are linearly independent, and hence $a_1 = b_1 = 0$. 

The same sort of argument shows that $a_j = b_j =0$ for each $j=1,2,3,4$. Thus the eight vectors $\Psi_j$, $j=1,2,3,4$ and 
$\Phi_j$, $j=1,2,3,4$ are linearly independent and in the nullspace of $\rho_{123}$, Since the dimension of $\C^2\otimes \C^2\otimes C^2$
is $8$, this means $\rho_{123} =0$, which is impossible. 

Therefore, the compatible pair of density matrices $\rho_{12}$ and $\rho_{23}$ defined by \eqref{def} has no common extension. 

Next, observe from \eqref{def} that 
$$\rho_1 = \sum_{j=1}^2 \mu_j |\eta_j\rangle\langle \eta_j|\ .$$
Since $\eta_1$ and $\eta_2$ are orthogonal, and since $\eta_1\otimes \psi_1$ and $\eta_2\otimes \psi_2$ are orthogonal
$$S_1 = S_{12} = -\sum_{j=1}^2 \mu_j \log \mu_j\ .$$

In the same manner we see that 
$$S_3 = S_{23} = -\sum_{j=1}^2 \nu_j \log \nu_j\ .$$
Thus we have
$$S_{12} + S_{23} = S_1 + S_3\ .$$

Finally, making  a sufficiently small change in $\{\eta_1,\eta_2\}$, we may arrange that this set of unit vectors is no longer orthogonal, but
that the   $8$ unit vectors $\Psi_j,\Phi_j$, $j=1,2,3,4$  defined in \eqref{psidef} and \eqref{phidef} are still linearly independent. 
Then the compatible pair $\rho_{12}$ and $\rho_{23}$ defined by \eqref{def}
has no common extension.  It is  still true that $S_{23} = S_3$.
Also, since $\{\psi_1,\psi_2\}$ is orthonormal, $\{\eta_1\otimes \psi_1,\eta_2\otimes \psi_2\}$ is orthonormal, and so 
$S_{12} = - \sum_{j=1}^2 \mu_j \log \mu_j$. However, $\rho_1=    \sum_{j=1}^2 \mu_j |\eta_j\rangle\langle \eta_j|$ and since 
$\{\eta_1,\eta_2\}$ is not orthogonal, 
$$S_1 <  -\sum_{j=1}^2 \mu_j\log \mu_j = S_{12}\ .$$
 Therefore $S_{12}+S_{23} > S_1+S_2$. 
\end{proof}

\begin{remark} While the condition that $\rho_1$ and $\rho_2$ have the same non-zero spectrum is necessary and sufficient 
to ensure 
that there exists a common purification of $\rho_1$ and $\rho_2$, the above construction shows that the condition that
 the nonzero spectrum of $\rho_{12}$ equals that of $\rho_{3}$, and that 
  the nonzero spectrum of $\rho_{23}$ equals that of $\rho_{1}$, together with compatibility,  does not ensure that there exists a common purification of 
  $\rho_{12}$ and $\rho_{23}$.
Indeed, let $\{\eta_1,\eta_2\}$, $\{\psi_1,\psi_2\}$, $\{\phi_1,\phi_2\}$ and $\{\chi_1,\chi_2\}$ be four orthonormal bases of $\C^2$
having no vectors (or their opposites) in common. Then with $\mu_j = \nu_j =1/2$, $j=1,2$,  (\ref{def}) defines
a compatible pair of density matrices with $\rho_{2} = \frac12 I$.  The non-zero spectrum of $\rho_{12}$, $\rho_{23}$,
$\rho_{1}$ and $\rho_{3}$ is $\{1/2,1/2\}$, and yet there is no common
extension of $\rho_{12}$ and $\rho_{23}$, as the proof of
Theorem~\ref{noext} shows.
\end{remark}

\end{document}